\documentclass[preprint,12pt]{elsarticle}

\usepackage{amsmath}
\usepackage{amssymb}
\usepackage{amsfonts}
\usepackage{amsthm}
\usepackage{times}
\usepackage{graphicx}
\usepackage{color}

\newtheorem{lemma}{Lemma}
\newtheorem{theorem}{Theorem}

\def\<{\leqslant}           
\def\>{\geqslant}           

\def\d{\partial}
\def\wh{\widehat}

\def\Re{{\rm Re}}   
\def\Im{{\rm Im}}   

\def\vec{{\rm vec}}   
\def\vech{{\rm vech}}   
\def\cH{{\cal H}}   
\def\mR{{\mathbb R}}    

\def\Tr{{\rm Tr}}       
\def\rT{{\rm T}}        

\def\bGamma{{\mathbf \Gamma}}
\def\bPi{{\mathbf \Pi}}

\def\bP{{\bf P}}    
\def\bE{{\bf E}}    

\def\[[[{[\![\![}
\def\]]]{]\!]\!]}

\def\bra{{\langle}}
\def\ket{{\rangle}}

\def\re{{\rm e}}        
\def\rd{{\rm d}}        

\def\fM{{\mathfrak M}}
\def\cL{{\mathcal L}}

\def\bC{{\mathbf C}}

\def\bR{{\mathbf R}}

\def\bJ{{\bf J}}

\def\br{{\bf r}}
\def\x{\times}
\def\ox{\otimes}

\def\cZ{{\mathcal Z}}

\def\bD{{\mathbf D}}

\def\bI{{\mathbf I}}

\def\cW{{\mathcal W}}

\def\cX{{\mathcal X}}

\def\cJ{{\mathcal J}}

\def\cM{{\cal M}}

\def\cC{{\cal C}}
\def\cR{{\cal R}}

\def\cI{{\mathcal I}}
\def\cP{{\mathcal P}}
\def\cQ{{\mathcal Q}}

\def\cA{{\cal A}}
\def\cB{{\cal B}}

\def\cS{{\mathcal S}}

\def\mS{{\mathbb S}}

\def\Ups{\Upsilon}

\journal{Systems and Control Letters}
\begin{document}

\begin{frontmatter}


\title{{\bf A quasi-separation principle and Newton-like scheme for coherent quantum LQG control}\footnotemark}
\author[address]{Igor G. Vladimirov}
\author[address]{\quad Ian R. Petersen}
\address[address]{School of Engineering and Information Technology,
        The University of New South Wales at the Australian Defence Force Academy,
        Canberra, ACT 2600, Australia,
        E-mail:
            {\tt igor.g.vladimirov@gmail.com} (I.G.Vladimirov),
            {\tt i.r.petersen@gmail.com} (I.R.Petersen).
}
\begin{abstract}
This paper is concerned with constructing an optimal controller in the coherent quantum Linear Quadratic Gaussian problem. A coherent quantum controller is itself a quantum system and is required to be physically realizable. The use of coherent control avoids the need for classical measurements, which inherently entail the loss of quantum information.  Physical realizability corresponds to the  equivalence of the controller to an open quantum harmonic oscillator and relates its state-space  matrices to the Hamiltonian, coupling and scattering operators of the oscillator. The Hamiltonian parameterization of the controller is combined with Frechet differentiation of the LQG cost with respect to the state-space matrices to obtain equations for the optimal controller. A quasi-separation principle for the gain matrices of the quantum controller is established, and a Newton-like iterative scheme for numerical solution of the equations is outlined.
\end{abstract}



\end{frontmatter}
\footnotetext{A shortened version of this work is to appear in the 18th IFAC World Congress Proceedings \cite{VP_2011}.}

\section{Introduction}

Sensitivity to observation is an inherent feature of quantum mechanical systems whose state is  affected by  interaction with a macroscopic measuring device. This motivates the use of coherent quantum controllers to replace the classical observation-actuation control loop by a measurement-free feedback, which is organized as an interconnection of the quantum plant with another quantum system.
If such a controller is implemented using quantum-optical components (for example, optical cavities and beam splitters) mediated by light fields \cite{GZ_2004}, then it is dynamically equivalent to an open quantum harmonic oscillator, which constitutes a building block of quantum systems described by linear quantum stochastic differential equations (QSDEs) \cite{P_1992,P_2010}.

This leads to the notion of  physical realizability which imposes quadra\-tic constraints on  the state-space matrices of the controller \cite{JNP_2008,NJP_2009,SP_2009}, thus complicating the solution of quantum control problems which are otherwise reduced to appropriate unconstrained problems for an equivalent classical system. The links between classical control problems and their quantum analogues are known, for example, for Linear Quadratic Gaussian (LQG) and $\cH_{\infty}$-control.

The Coherent Quantum LQG (CQLQG) problem seeks a physically realizable quantum controller to minimize the average output ``energy'' of the closed-loop system per unit time.  This problem has been addressed in  \cite{NJP_2009}, where a numerical procedure was proposed for  finding \textit{suboptimal} controllers to ensure a given upper bound on the LQG cost. Instead, the present paper  focuses on necessary conditions for optimality and second order conditions for local strict optimality of a physically realizable controller and computation of the \textit{optimal} controller. Both approaches make use of the fact that the CQLQG problem is equivalent  to a constrained LQG problem for a classical plant, with the LQG cost computed as the squared $\cH_2$-norm of the system in terms of the controllability and observability Gramians satisfying algebraic Lyapunov equations.

We utilize a Hamiltonian parameterization that relates the state-space matrices of a physically realizable controller to the free Hamiltonian, coupling and scattering operators of an open quantum harmonic oscillator \cite{EB_2005}. To obtain equations for the optimal quantum controller, we employ an algebraic approach,
based on the Frechet differentiation of the LQG cost with respect to the   state-space  matrices from \cite{VP_2010} and similar to \cite{SIG_1998}.
The resulting  equations for the optimal controller involve the inverse of special self-adjoint operators on matrices  that requires the use of vectorization \cite{M_1988}. Their spectral properties play an important role in the present study.

Although the optimal CQLQG controller does not inherit the control/filtering separation principle of the classical LQG control problem, a partial decoupling of equations for the gain matrices still holds. This \textit{quasi-separation} property leads to a Newton-like scheme for numerical computation of the quantum controller that involves the second order Frechet derivative of the LQG cost which is related to the perturbation of solutions to algebraic Lyapunov equations.

The paper is organised as follows.  Section~\ref{sec:plant} specifies the quantum plants being  considered. Sections~\ref{sec:controller} and \ref{sec:PR} describe physically realizable   quantum controllers. Section~\ref{sec:problem} formulates the CQLQG control problem. Sections~\ref{sec:bGamma} and \ref{sec:class} introduce auxiliary classes of matrices and  self-adjoint operators. Section~\ref{sec:optimal} obtains equations for the optimal CQLQG controller. Section~\ref{sec:separation} discusses the quasi-separation property. Section~\ref{sec:d2EdR2}  establishes a second order condition of optimality. Section~\ref{sec:newton} outlines a Newton-like scheme for computing the optimal controller. Appendices 
provide a subsidiary material on invertibility of the special self-adjoint operators, perturbations of inverse Lyapunov operators and Frechet differentiation of the LQG cost.

\section{Quantum plant}\label{sec:plant}

We consider a quantum plant with an $n$-dimensional state vector
$x_t$, a $p$-dimensional output $y_t$ and inputs $w_t$, $\eta_t$
of dimensions $m_1$, $m_2$. The state and the
output are governed by the QSDEs:
\begin{align}
\label{x}
    \rd x_t
    & =
    A x_t\rd t  +  B_1 \rd w_t + B_2 \rd \eta_t,\\
\label{y}
    \rd y_t
    & =
    z_t\rd t  +  D \rd w_t,\\
\label{z}
    z_t
    & =
     Cx_t.
\end{align}
Here,
$
    A\in \mR^{n\x n}
$,
$
    B_k\in \mR^{n\x m_k}
$,
$
    C\in \mR^{p\x n}
$,
$
    D\in \mR^{p\x m_1}
$
are constant matrices, and $z_t$ is a ``signal part'' of $y_t$. The state dimension $n$
and the input dimensions $m_1$, $m_2$ are even:
$
    n = 2\nu$,
    $
    m_k = 2\mu_k
$.
 The plant state vector $x_t$
is formed by self-adjoint operators
(similar to the position and
momentum operators) and, in the Heisenberg picture of
quantum mechanics, evolves in time $t$. The entries of the
$m_1$-dimensional vector $w_t$ are self-adjoint quantum Wiener
processes  \cite{P_1992} whose infinitesimal increments compose with each other  according to the Ito table
\begin{equation}
\label{Ftable}
        \rd w_t \rd w_t^{\rT}
    =
    F\rd t.
\end{equation}
Here, $F$ is a complex positive semi-definite Hermitian matrix which, on the right-hand side of (\ref{Ftable}), is a shorthand notation for $F\ox \cI$, with $\cI$ the identity operator on the underlying boson Fock space and $\ox$ the tensor product. We
assume that vectors are organized as columns unless indicated
otherwise, and the transpose $(\cdot)^{\rT}$ acts on vectors and
matrices with operator-valued entries as if the latter were scalars.
Also, $(\cdot)^{\dagger}:= ((\cdot)^{\#})^{\rT}$ denotes the transpose
of the entry-wise adjoint $(\cdot)^{\#}$. Associated with the Hermitian  matrix $F$ from (\ref{Ftable}) are real matrices
$
    S
    :=
    (
        F + \overline{F}
    )/2
     =
    \Re F
$  and  $
    T
     :=
    (
        F - \overline{F}
    )/i
     =
    2\Im F
$,
where $\overline{(\cdot)}$, $\Re(\cdot)$ and $\Im(\cdot)$ are the entry-wise complex conjugate, real and imaginary parts, and $i :=\sqrt{-1}$ is the imaginary unit.  The symmetric matrix $S$ contributes to the evolution of the covariance matrix of the plant state vector $x_t$, whilst $T$ is antisymmetric and affects the cross-commutations between the entries of $x_t$ through
$
        [\rd w_t, \rd w_t^{\rT}]
    :=
        \rd w_t \rd w_t^{\rT}
        -
        (\rd w_t \rd w_t^{\rT})^{\rT}
    =
    (F-F^{\rT})\rd t
    =
    i T\rd t
$.
Here, the commutator $[\alpha, \beta] := \alpha\beta-\beta\alpha$  applies entry-wise, and the relation $F^{\rT} = \overline{F}$ is ensured by $F = F^*$. In what follows, it is assumed that $S = I_{m_1}$, and $T$ is canonical in the sense that
\begin{equation}
\label{Tcanonical}
    T
    :=
    I_{\mu_1}
    \ox \bJ,
    \qquad
    \bJ
    :=
    \begin{bmatrix}
    0 & 1\\
    -1 & 0
    \end{bmatrix},
\end{equation}
where $I_r$ is the identity matrix of order $r$. That is, $T$ is a block diagonal matrix with
$\mu_1$ copies of $\bJ$ over the diagonal. By permuting the rows and
columns, the matrix $T$ from (\ref{Tcanonical}) can be brought to an
equivalent canonical form
\begin{equation}
\label{Talternative}
    T
    =
    \bJ\ox I_{\mu_1}
    =
    \begin{bmatrix}
        0_{\mu_1} & I_{\mu_1}\\
        -I_{\mu_1} & 0_{\mu_1}
    \end{bmatrix},
\end{equation}
where $0_r$ denotes the $(r\x r)$-matrix of zeros. The canonical
antisymmetric matrix $J$ of any order satisfies $J^2=-I$. Quantum Wiener processes will be assumed to have the canonical Ito matrix $F=I+iJ/2$.


\section{Coherent quantum controller}\label{sec:controller}

A measurement-free coherent quantum controller is another quantum system with a $n$-dimensional state vector $\xi_t$ with self-adjoint operator-valued entries whose interconnection with the plant (\ref{x})--(\ref{z}) is described by QSDEs
\begin{align}
\label{xi}
    \rd \xi_t
     & =
    a\xi_t\rd t + b_1 \rd \omega_t + b_2\rd y_t,\\
\label{eta}
    \rd \eta_t
     & =
    \zeta_t \rd t + \rd \omega_t,\\
\label{zeta}
    \zeta_t
    & =
    c\xi_t.
\end{align}
Here,
$
    a \in \mR^{n\x n}
$,
$
    b_1\in \mR^{n\x m_2}
$,
$
    b_2\in \mR^{n\x p}
$,
$
    c\in \mR^{m_2\x n}
$, and $\omega_t$ is a
$m_2$-dimensional vector of self-adjoint quantum Wiener processes
which commute with the plant noise $w_t$ in (\ref{x}) and (\ref{y}).
The combined set of equations
 (\ref{x})--(\ref{z}) and (\ref{xi})--(\ref{zeta})
describes the fully quantum closed-loop system in Fig.~\ref{fig:system},
\begin{figure}[htb]
\begin{center}
\unitlength=1.6mm
\begin{picture}(50.00,20.00)
    \put(10,10){\framebox(10,10)[cc]{{\small plant}}}
    \put(30,6){\framebox(10,10)[cc]{{\small controller}}}
    \put(10,15){\line(-1,0){10}}
    \put(0,15){\line(0,-1){15}}
    \put(0,0){\line(1,0){50}}
    \put(50,0){\line(0,1){7}}
    \put(50,7){\vector(-1,0){10}}
    \put(50,15){\vector(-1,0){10}}
    \put(30,11){\vector(-1,0){10}}
    \put(30,19){\vector(-1,0){10}}
    \put(25,20){\makebox(0,0)[cb]{$w$}}
    \put(25,10){\makebox(0,0)[ct]{$\eta$}}
    \put(45,16){\makebox(0,0)[cb]{$\omega$}}
    \put(45,6){\makebox(0,0)[ct]{$y$}}
\end{picture}\vskip-1mm
\caption{
    The quantum closed-loop system described by  (\ref{x})--(\ref{z}) and (\ref{xi})--(\ref{zeta}), where the plant and controller noises $w$ and $\omega$ are
    commuting quantum Wiener processes.
}
\label{fig:system}
\end{center}
\end{figure}
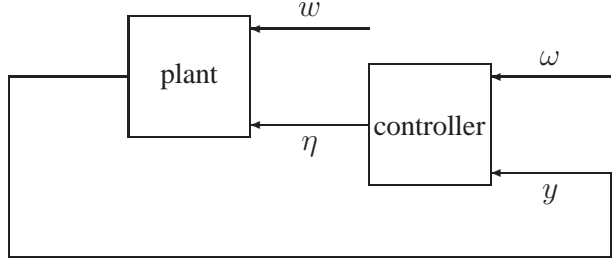
whose output
observables form a $p_0$-dimensional process
\begin{equation}
\label{cZ}
    \cZ_t
     =
    C_0 x_t  + D_0 \zeta_t,
\end{equation}
where $C_0\in \mR^{p_0\x n}$ and $D_0\in \mR^{p_0\x m_2}$ are given matrices. The $2n$-dimensional combined state vector $\cX_t:=
[x_t^{\rT}\ \xi_t^{\rT}]^{\rT}$ and the output $\cZ_t$ of the
closed-loop system  are governed by the QSDEs
\begin{equation}
\label{closed}
    \rd \cX_t
      =
      \cA     \cX_t\rd t +   \cB      \rd \cW_t,
    \qquad
    \cZ_t
     =
      \cC       \cX_t.
\end{equation}
Here, the combined quantum Wiener process
$\cW_t:= [w_t^{\rT}\ \omega_t^{\rT}]^{\rT}$ has a block diagonal Ito table. The matrices $  \cA    $, $  \cB $, $
 \cC      $ of the closed-loop system (\ref{closed}) are given by
\begin{equation}
\label{cABC}
    \left[
        \begin{array}{c|c}
              \cA     &   \cB     \\
            \hline
              \cC       & 0
        \end{array}
    \right]
    =
    \left[
        \begin{array}{cc|cc}
                  A & B_2 c   & B_1 & B_2\\
                  b_2C & a & b_2D & b_1\\
            \hline
                  C_0 & D_0 c &  0 & 0
        \end{array}
    \right]
    =
    \left[
        \begin{array}{cc|c}
                  A & B_2 c   & B\\
                  b\bC & a & b\bD\\
            \hline
                  C_0 & D_0 c &  0
        \end{array}
    \right],
\end{equation}
where
\begin{equation}
\label{bB_CD}
    b
    :=
        \begin{bmatrix}
            b_1 & b_2
        \end{bmatrix},\quad
    B
    :=
        \begin{bmatrix}
            B_1 & B_2
        \end{bmatrix},\quad
    \bC
    :=
        \begin{bmatrix}
            0 \\
            C
        \end{bmatrix},\quad
    \bD
    :=
        \begin{bmatrix}
            0 & I \\
            D & 0
        \end{bmatrix}.
\end{equation}
The dependence of $\cA$, $\cB$, $\cC$ on the controller matrices
$a$, $b$, $c$ is equivalently described by
\begin{equation}
\label{Gamma}
    \Gamma
     :=
        \begin{bmatrix}
              \cA     &   \cB     \\
              \cC       & 0
        \end{bmatrix}
    =
    \Gamma_0
    +
    \Gamma_1
    \gamma
    \Gamma_2,
\qquad
    \gamma
    :=
        \begin{bmatrix}
            a & b\\
            c & 0
        \end{bmatrix}.
\end{equation}
The affine map $\gamma \mapsto \Gamma$ is completely specified by
the plant (\ref{x})--(\ref{z})  through
the matrices
\begin{equation}
\label{Gamma0_Gamma12}
    \Gamma_0
    :=
        \begin{bmatrix}
            A & 0 &  B\\
            0 & 0_n & 0\\
            C_0 & 0 & 0
        \end{bmatrix},\qquad
    \Gamma_1
    :=
        \begin{bmatrix}
            0 & B_2\\
            I_n & 0\\
            0 & D_0
        \end{bmatrix},
    \qquad
    \Gamma_2
    :=
        \begin{bmatrix}
            0 & I_n & 0\\
            \bC & 0 & \bD
        \end{bmatrix}.
\end{equation}
Using the terminology introduced formally  in Section~\ref{sec:class}, the map $\gamma \mapsto \Gamma_1\gamma \Gamma_2$ in (\ref{Gamma}) is a grade one linear operator $\[[[\Gamma_1, \Gamma_2\]]]$.

\section{Physical realizability}\label{sec:PR}

A controller (\ref{xi})--(\ref{zeta})
is called \textit{physically realizable} (PR) \cite{JNP_2008,NJP_2009},  if its state-space matrices satisfy
\begin{equation}
\label{PR1_PR2}
    aJ_0 +J_0a^{\rT}
    +
    b
    J
    b^{\rT}
    =0,
    \qquad
    b_1 = J_0 c^{\rT} J_2.
\end{equation}
Here, $J$ is a block-diagonal matrix, partitioned in conformance
with the matrix $b$ from (\ref{bB_CD}) as
\begin{equation}
\label{J}
   J
   :=
   \bD
   \begin{bmatrix}
    J_1 & 0\\
    0 & J_2
   \end{bmatrix}
   \bD^{\rT}
   =
   \begin{bmatrix}
    J_2 & 0\\
    0 & DJ_1D^{\rT}
   \end{bmatrix},
\end{equation}
and $J_0$, $J_1$, $J_2$ are fixed real antisymmetric matrices of orders $n$, $m_1$, $m_2$, which specify the commutation relations for the controller state variables $\xi_t$ and the plant and controller noises $w$ and $\omega$. For convenience, $J_0$, $J_1$, $J_2$ are assumed to have the canonical form (\ref{Tcanonical}) or (\ref{Talternative}).
%
The relations (\ref{PR1_PR2}) describe the equivalence of the controller to an open quantum
harmonic oscillator and the possibility of its quantum optical implementation \cite{GZ_2004}.  The first of these equations is the condition for preservation of the canonical commutation relations for the state variables of the quantum harmonic oscillator. The second  PR condition, which relates the matrices $b_1$ and $c$ by a linear bijection, describes the unitary transformation of the quantum Wiener process at the input of the quantum harmonic oscillator.
The first of the PR conditions  (\ref{PR1_PR2}), which is a linear equation with respect to $a$,   determines $a$ as a quadratic function of $b$ up to the subspace of Hamiltonian matrices
$
    \{
        a \in \mR^{n\x n}:\
        a J_0 + J_0 a^{\rT} = 0
    \}=J_0 \mS_n=\mS_n J_0
$, with  $\mS_n$ the subspace of real symmetric matrices of order $n$:
\begin{equation}
\label{a}
    a
    =
    \underbrace{J_0 R}_{\rm Hamiltonian\ matrix}
    +
    \underbrace{
    b
    J
    b^{\rT}J_0/2.}_{\rm particular\ solution}
\end{equation}
Here, $R\in \mS_n$ specifies the
free Hamiltonian operator $\xi_t^{\rT} R \xi_t/2$ of the quantum harmonic oscillator
\cite[Eqs. (20)--(22) on pp. 8--9]{EB_2005}. Since the matrix $bJb^{\rT}$ is antisymmetric, $bJb^{\rT}J_0$ is skew-Hamiltonian. Therefore, (\ref{a}) describes an orthogonal decomposition of the matrix $a$ into projections onto the subspaces of Hamiltonian and skew-Hamiltonian matrices in the sense of the Frobenius inner product of real matrices
$
    \bra X, Y\ket
    :=
    \Tr(X^{\rT}Y)
$, with $\|X\|:= \sqrt{\bra X, X \ket}$ the Frobenius norm. From the second PR condition in (\ref{PR1_PR2}) and the canonical structure of $J_0$ and $J_2$, it follows that the matrix $c$ is related to $b_1$ by
\begin{equation}
\label{c}
    c
    =
    J_2 b_1^{\rT}J_0
    =
    J_2 \bI^{\rT}b^{\rT}J_0,
    \qquad
    \bI:=
    \begin{bmatrix}
        I\\
        0
    \end{bmatrix},
\end{equation}
where, in view of (\ref{bB_CD}),  the matrix $\bI$ ``extracts'' $b_1$ from $b$  as $    b_1 = b \bI $. In combination with the decomposition (\ref{a}), this implies that, for a physically realizable quantum controller,  the matrix $\gamma$ in (\ref{Gamma})  is completely parameterized by the matrices $R$ and $b$ as
\begin{equation}
\label{gamma}
    \gamma
    =
    \begin{bmatrix}
      J_0 R + bJb^{\rT}J_0/2        & b \\
      J_2 \bI^{\rT} b^{\rT}J_0      & 0 \\
    \end{bmatrix}.
\end{equation}
In view of the physical meaning of $R$, we will refer to (\ref{gamma}) as the \textit{Hamiltonian parameterization} of the coherent quantum controller, with the $\mS_n \x \mR^{n\x (m_2+p)}$-valued parameter
$
    \begin{bmatrix}
        R & b
    \end{bmatrix}
$; see Fig.~\ref{fig:PR}.
\begin{figure}[htb]
\begin{center}
\unitlength=1.5mm
\begin{picture}(30.00,27.00)

\multiput(0,0)(10,-10){2}{
\qbezier(6.5, 25)(7, 27)(9, 27.5)

\qbezier(21, 27.5)(23, 27)(23.5, 25)
\multiput(9.5, 27.5)(2, 0){6}{\line(1, 0){1}}
\multiput(6.5, 24.5)(0, -2){6}{\line(0, -1){1}}
\qbezier(6.5, 13)(7, 11)(9, 10.5)
\multiput(9.5, 10.5)(2, 0){1}{\line(1, 0){1}}
\qbezier(13.5, 13)(13, 11)(11, 10.5)
\qbezier(13.5, 18)(14, 20)(16, 20.5)
\multiput(13.5, 13.5)(0, 2){2}{\line(0, 1){1}}
\qbezier(23.5, 23)(23, 21)(21, 20.5)
\multiput(23.5, 23.5)(0, 2){1}{\line(0, 1){1}}
\multiput(20.5, 20.5)(-2, 0){2}{\line(-1, 0){1}}}



    \put(10,19){\makebox(0,0)[cc]{{\small $\gamma$}}}
    \put(20,9){\makebox(0,0)[cc]{{\small$\Gamma$}}}

    \put(0,24){\circle{4}}
    \put(0,24){\circle{3}}
    \put(0,24){\makebox(0,0)[cc]{{\small $R$}}}
    \put(2,24.2){\vector(1,0){6}}
    \put(2,23.8){\vector(1,0){6}}

    \put(10,24){\circle{4}}
    \put(10,24){\makebox(0,0)[cc]{{\small $a$}}}
    \put(18.35,22.65){\vector(-1,-1){7}}
    \put(18.65,22.35){\vector(-1,-1){7}}

    \put(20,22){\vector(0,-1){6}}

    \put(11.5,22.5){\vector(1,-1){7}}

    \put(21.5,22.5){\vector(1,-1){7}}
    \put(11.5,12.5){\vector(1,-1){7}}

    \put(21.5,12.5){\vector(1,-1){7}}

    \put(20,24){\circle{4}}
    \put(20,24){\circle{3}}
    \put(20,24){\makebox(0,0)[cc]{{\small$b$}}}

    \put(18,24.2){\vector(-1,0){6}}
    \put(18,23.8){\vector(-1,0){6}}

    \put(10,14){\circle{4}}
    \put(10,14){\makebox(0,0)[cc]{{\small$c$}}}
    \put(12,14){\vector(1,0){6}}

    \put(20,14){\circle{4}}
    \put(20,14){\makebox(0,0)[cc]{{\small$\cA$}}}

    \put(30,12){\vector(0,-1){6}}

    \put(22,4){\vector(1,0){6}}

    \put(20,4){\circle{4}}
    \put(20,4){\makebox(0,0)[cc]{{\small$\cC$}}}

    \put(30,14){\circle{4}}
    \put(30,14){\makebox(0,0)[cc]{{\small$\cB$}}}

    \put(30,4){\circle{4}}
    \put(30,4){\makebox(0,0)[cc]{{\small$ E $}}}
\end{picture}\vskip-2mm
\caption{
    This directed acyclic graph describes the dependence of the LQG cost $ E $ of the closed-loop system on the matrices
    $R$ and $b$. An oriented edge $\bigcirc\!\!\!\!\!\tiny{\alpha}\!\!\rightarrow\!\!\bigcirc\!\!\!\!\!\tiny{\beta}$ signifies ``$\beta$ depends on $\alpha$''. The dashed lines encircle the matrix triples  $\gamma$ and $\Gamma$ defined by (\ref{Gamma}).
    The emergence of $R$ and the dependencies indicated by double arrows represent the PR conditions for the quantum controller, with $a$, $b$, $c$ being otherwise independent.
}
\label{fig:PR}
\end{center}
\end{figure}
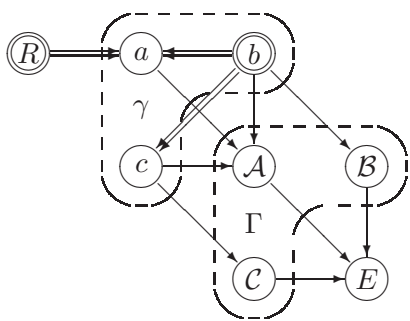
 The PR conditions (\ref{PR1_PR2}) are invariant under the group of similarity transformations of the controller matrices
$    (a,b,c)
    \mapsto
    (\sigma a \sigma^{-1}, \sigma b, c\sigma^{-1})
$,
where $\sigma$ is any real symplectic matrix of order $n$ (that is, $\sigma J_0\sigma^{\rT} = J_0$). This corresponds to the canonical state transformation $\xi_t \mapsto \sigma \xi_t$; see also \cite[Eqs. (12)--(14)]{S_2000}.  Any such transformation of a physically realizable controller leads to its  equivalent state-space representation, with the matrix $R$ transformed as $R\mapsto \sigma^{-\rT}R \sigma^{-1}$.

\section{Coherent quantum LQG control problem}\label{sec:problem}

The Coherent Quantum LQG (CQLQG) control problem  \cite{NJP_2009} consists in minimizing the average output ``energy'' of the closed-loop system (\ref{closed}):
\begin{align}
\nonumber
     E
     :=&
    \lim_{t\to +\infty}
    \left(
        \frac{1}{t}
        \int_{0}^{t}
        \bE
        (
            \cZ_s^{\rT} \cZ_s
        )
        \rd s
    \right)
        =
    \Tr(\cC       P   \cC      ^{\rT})\\
\label{E}
    =&
    \Tr(  \cB^{\rT}       Q   \cB)
    =
    -2
    \bra \cA, H\ket
    \longrightarrow
    \min.
\end{align}
The minimum is taken over the  $n$-dimensional controllers (\ref{xi})--(\ref{zeta}) which make the matrix $ \cA $ in (\ref{cABC}) Hurwitz and satisfy the PR conditions (\ref{PR1_PR2}). Here, $\bE X:= \Tr(\rho X)$ denotes  the quantum expectation over the underlying density operator $\rho$, and
$
    P
    :=
    \lim_{t\to +\infty}
    \Re\bE
    (
        \cX_t\cX_t^{\rT})
$
is the steady-state covariance matrix of the state vector of the closed-loop system. Also,
we use the shorthand notation
\begin{equation}
\label{H}
    H:= QP,
\end{equation}
with $ P $ and $ Q $
satisfying the algebraic Lyapunov equations
\begin{equation}
\label{PQ}
      \cA     P
    +
     P   \cA    ^{\rT}
    +
      \cB       \cB    ^{\rT}
    = 0,
\quad
      \cA    ^{\rT} Q
    +
     Q    \cA
    +
      \cC      ^{\rT}  \cC
    = 0,
\end{equation}
so that these matrices are the controllability and observability
Gramians of the state-space realization triple $( \cA    ,  \cB     ,  \cC      )$. The spectrum of the diagonalizable matrix $H$ in (\ref{H}) is
formed by the squared Hankel singular values of the system, and we will refer to $H$ as the \textit{Hankelian}.
The fact that $ E $ coincides with the squared $\cH_2$-norm of a classical strictly proper linear time invariant system enables the CQLQG problem (\ref{E}) to be recast as a constrained LQG control problem for an equivalent classical plant
\begin{equation}
\label{class_plant}
    \left[
        \begin{array}{l|lc}
                  A     &   B   & B_2\\
                \hline
                  C_0   &   0   & D_0\\
                  \bC   &  \bD  & 0
        \end{array}
    \right]
    =
    \left[
        \begin{array}{l|llc}
                  A     &   B_1 & B_2   & B_2\\
                \hline
                  C_0   &   0   & 0     & D_0\\
                  0     &   0   & I     & 0\\
                  C   &  D  & 0 & 0
        \end{array}
    \right]
\end{equation}
driven by a $(m_1+m_2)$-dimensional standard Wiener process, with the controller being noiseless.
%
We will employ the smooth dependence of the cost $ E $ on the matrices $R$ and $b$ which govern the Hamiltonian parameterization (\ref{gamma}) of a physically realizable stabilizing controller. The conditions of  optimality, obtained in Section~\ref{sec:optimal},   utilize the Frechet differentiation of the LQG cost with respect to the state-space realization matrices \cite{VP_2010} assembled into matrices with a specific sparsity pattern and an auxiliary class of self-adjoint operators introduced in Sections~\ref{sec:bGamma} and \ref{sec:class}.


\section{The $\Gamma$ sparsity structure}\label{sec:bGamma}

The subsequent considerations involve Frechet differentiation with respect to state-space realization matrices assembled into matrices of the ``$\Gamma$-shaped'' sparsity  structure  (\ref{Gamma}). We denote by
\begin{equation}
\label{bGamma}
    \bGamma_{r,m,p}
    \!:=\!
    \left\{
                 \begin{bmatrix}
                     \varphi  &  \sigma \\
                     \tau & 0
                 \end{bmatrix}:
         \varphi \in \mR^{r\x r},
         \sigma \in \mR^{r\x m},
         \tau \in \mR^{p\x r}
    \right\}\!\!
\end{equation}
the Hilbert space of real $(r+p)\x(r+m)$-matrices
whose bottom-right block of size $(p\x m)$ is zero.
The space $\bGamma_{r,m,p}$, which is a subspace of $\mR^{(r+p)\x(r+m)}$,
inherits the Frobenius inner product of matrices. Let
$\bPi_{r,m,p}$ denote the orthogonal projection onto
$\bGamma_{r,m,p}$ whose action on a $(r+p)\x (r+m)$-matrix consists
in padding its bottom-right $(p\x m)$-block $\psi$ with zeros:
\begin{equation}
\label{bPi}
    \bPi_{r,m,p}
    \left(
                 \begin{bmatrix}
                     \varphi  &  \sigma \\
                     \tau & \psi
                 \end{bmatrix}
    \right)
    =
        \begin{bmatrix}
            \varphi  &  \sigma \\
            \tau & 0
        \end{bmatrix}.
\end{equation}
The subscripts in $\bGamma_{r,m,p}$ and $\bPi_{r,m,p}$ will often be omitted
for brevity. The Frechet derivative $\d_X f$ of a smooth function
$
    \bGamma
     \ni
    \begin{bmatrix}
        \varphi & \sigma\\
        \tau & 0
    \end{bmatrix}
     =:
    X
     \mapsto
    f(X)
     \in
    \mR
$
belongs to the same Hilbert space (\ref{bGamma}) and inherits the sparsity structure:
$
    \d_X f
     =
    \begin{bmatrix}
        \d_{\varphi} f & \d_{\sigma} f\\
        \d_{\tau}f & 0
    \end{bmatrix}
$.

\section{Special self-adjoint operators}\label{sec:class}

For the purposes of Section~\ref{sec:optimal}, we associate a linear operator $\[[[\alpha, \beta\]]]: \mR^{p\x q}\to \mR^{s\x t}$ with a pair of matrices $\alpha \in \mR^{s \x p}$ and  $\beta\in \mR^{q\x t}$, by
\begin{equation}
\label{cL}
    \[[[\alpha, \beta\]]](X)
    :=
    \alpha X \beta.
\end{equation}
The map $(\alpha, \beta)\mapsto \[[[\alpha, \beta\]]]$ from the direct product of the matrix spaces to the space of linear operators on matrices is bilinear. If $s=p$ and $t=q$, then the spectrum of the operator $\[[[\alpha, \beta\]]]$ on $\mR^{p\x q}$ consists of the pairwise products $\lambda_j\mu_k$ of the eigenvalues $\lambda_1, \ldots, \lambda_p$ and $\mu_1, \ldots, \mu_q$ of the matrices $\alpha$ and $\beta$, so that their spectral radii  are related by
\begin{equation}
\label{rrr}
    \br(\[[[\alpha, \beta \]]])
    =
    \br(\alpha)
    \br(\beta).
\end{equation}
Furthermore, for any positive integer $r$ and matrices $\alpha_1, \ldots, \alpha_r$ $ \in \mR^{s\x p}$ and $\beta_1, \ldots, \beta_r\in \mR^{q\x t}$, we define a linear operator
\begin{equation}
\label{cLsum}
    \[[[
        \alpha_1,\beta_1
        \mid
        \ldots
        \mid
        \alpha_r, \beta_r
    \]]]
    :=
    \sum_{k=1}^{r}
    \[[[
        \alpha_k,  \beta_k
    \]]],
\end{equation}
where the matrix pairs are separated by ``$\mid$''s.
Of importance will be
self-adjoint linear operators on the Hilbert space
$\mR^{p\x q}$  of the form (\ref{cLsum}) where
$\alpha_1, \ldots, \alpha_r\in \mR^{p\x p}$ and $\beta_1, \ldots, \beta_r\in \mR^{q\x
q}$ are such that for any $k=1, \ldots, r$, the matrices $\alpha_k$ and $\beta_k$ are either both symmetric or both antisymmetric. Such an operator  (\ref{cLsum}) will be referred to as a \textit{self-adjoint operator of grade} $r$.
The self-adjointness is understood in the sense of the Frobenius inner product on $\mR^{p\x q}$ and follows
from the property that, in each of the cases
$(\alpha^{\rT},\beta^{\rT}) =(\pm\alpha, \pm\beta)$, the adjoint $\[[[ \alpha, \beta\]]]^{\dagger} = \[[[\alpha^{\rT}, \beta^{\rT}\]]]$ coincides with $\[[[\alpha, \beta\]]]$. In these cases,  as for any self-adjoint operator, the eigenvalues of $\[[[\alpha, \beta\]]]$ are all real.
\begin{lemma}
\label{lem:spec}
If $\alpha \in \mR^{p\x p}$ and $\beta\in \mR^{q\x q} $ are both antisymmetric, then the spectrum of $\[[[\alpha, \beta\]]]$ is symmetric about the origin. If $\alpha$ and $\beta $ are both symmetric and positive (semi-) definite, then $\[[[\alpha, \beta\]]]$ is positive (semi-) definite, respectively.
\end{lemma}
\begin{proof}
If $\alpha$ and $\beta$ are both antisymmetric, then their eigenvalues $\lambda_1, \ldots, \lambda_p$ and $\mu_1, \ldots, \mu_q$ are all pure imaginary and symmetric about the origin \cite{HJ_2007}. Hence, the eigenvalues $\lambda_j\mu_k$ of $\[[[\alpha, \beta\]]]$ also form a set  which is symmetric about the origin. By a similar reasoning, if $\alpha$ and $\beta$ are real positive (semi-) definite symmetric matrices, then their eigenvalues are all real and (nonnegative) positive, and hence, so are the eigenvalues of $\[[[\alpha, \beta\]]]$ which implies its positive (semi-) definiteness. Alternatively, the second assertion of the lemma also follows from the relation $\[[[\alpha,\beta\]]] = \[[[\sqrt{\alpha}, \sqrt{\beta}\]]]^2$ which holds for any positive semi-definite symmetric matrices $\alpha \in \mR^{p\x p}$ and $ \beta\in \mR^{q\x q}$, so that $\bra X, \alpha X \beta \ket = \|\sqrt{\alpha} X \sqrt{\beta}\|^2\> 0$ for any $X\in \mR^{p\x q}$.
\end{proof}
Whilst the operator (\ref{cL})  with nonsingular $\alpha$ and
$\beta$ is straightforwardly
invertible: $\[[[\alpha,  \beta\]]]^{-1} = \[[[\alpha^{-1},\beta^{-1}\]]]$, the
inverse of
$\cM :=     \[[[
        \alpha_1,\beta_1
        \mid
        \ldots
        \mid
        \alpha_r, \beta_r
    \]]]
$
from (\ref{cLsum}) for $r>1$ (except for the case
$\sum_{j,k} \[[[\alpha_j, \beta_k\]]] = \[[[\sum_j\alpha_j, \sum_k\beta_k\]]]$,
which reduces to a grade one operator, or special Lyapunov operators $\[[[\alpha, I\]]] + \[[[ I,\alpha\]]]$ with $\alpha=\alpha^{\rT}$ which are treated by diagonalizing  the matrix $\alpha$), can only be computed
using the vectorization of matrices \cite{M_1988} as
$    \cM^{-1}(Y)
    =
    \vec^{-1}
    (
        \Xi^{-1}
    \vec(Y)
    )
$,
provided that the matrix
$    \Xi
    :=
    \sum_{k=1}^{r}
    \beta_k^{\rT}\ox \alpha_k
$ is nonsingular.
Here, $\vec:\mR^{p\x q}\to \mR^{pq}$ is a linear bijection which
maps a matrix $X$ to the vector obtained by writing the columns
$X_{\bullet1}, \ldots, X_{\bullet q}$ of the matrix  one underneath
the other.  Invertibility conditions for grade two operators    is discussed in \ref{sec:inversion}.

\section{Equations for the optimal controller}\label{sec:optimal}

 Necessary conditions for optimality in the class of \linebreak $n$-dimensional physically realizable stabilizing controllers are obtained by equating the Frechet derivatives of the LQG cost $E$  with respect to $R$ and $b$ to zero. In view of Fig.~\ref{fig:PR}, the chain rule allows  the differentiation to be carried out in three steps.   First, the matrices $\cA$, $\cB$, $\cC$ of the closed-loop system are considered to be independent variables. Below is an adaptation of \cite[Lemma~7  of Appendix~B]{VP_2010}
 whose proof is given to make the exposition self-contained.

\begin{lemma}
\label{lem:dEdGamma}
Suppose the matrix $\cA$ in (\ref{cABC}) is Hurwitz. Then the Frechet derivative of the LQG cost $ E $ from (\ref{E})   with respect to the matrix $\Gamma$ from (\ref{Gamma}) is
\begin{equation}
\label{dEdGamma}
    \d_{\Gamma} E
    =
    2
    \begin{bmatrix}
        H & Q\cB\\
        \cC P & 0
    \end{bmatrix}.
\end{equation}
Here, $H$ is the Hankelian defined by (\ref{H}) in terms of  the Gramians $P$, $Q$ from
 (\ref{PQ}).
\end{lemma}
\begin{proof}
As discussed in Section~\ref{sec:bGamma}, the Frechet derivative $\d_{\Gamma} E $ inherits the block structure of the matrix $\Gamma$:
\begin{equation}
\label{dEdGamma_blocks}
    \d_{\Gamma} E
    =
    \begin{bmatrix}
        \d_{\cA} E  & \d_{\cB} E \\
        \d_{\cC} E  & 0
    \end{bmatrix}.
\end{equation}
We will now compute the blocks of this matrix. To calculate $\d_{\cA} E $, let $\cB$ and $\cC$ be fixed. Then the first variation of $ E $ with respect to $\cA$ is
$
    \delta  E
    =
    \bra
        \cC^{\rT}\cC,
        \delta P
    \ket
    =
    -
    \bra
        \cA^{\rT}Q+Q\cA,
        \delta P
    \ket
    =
    -
    \bra
        Q,
        \cA\delta P +(\delta P)\cA^{\rT}
    \ket
    =
    \bra
        Q,
        (\delta \cA) P +P\delta\cA^{\rT}
    \ket
    =
    2\bra
        H, \delta\cA
    \ket
$,
which implies that
\begin{equation}
\label{dEdA}
    \d_{\cA} E  = 2H.
\end{equation}
To compute $\d_{\cB} E $, suppose $\cA$ and $\cC$ are fixed. Then the observability Gramian $Q$, which is a function of $\cA$ and $\cC$, is also constant, and  the first variation of $ E $ with respect to $\cB$ is
$
    \delta  E
    =
    \bra
        Q,
        \delta(\cB\cB^{\rT})
    \ket
    =
    \bra
        Q,
        (\delta\cB)\cB^{\rT}+\cB\delta\cB^{\rT}
    \ket
    =
    2
    \bra
        Q\cB,
        \delta\cB
    \ket
$,
and hence,
\begin{equation}
\label{dEdB}
    \d_{\cB} E  = 2Q\cB.
\end{equation}
The derivative $\d_{\cC} E $ is calculated by a similar reasoning. Assuming $\cA$ and $\cB$ (and so also the controllability Gramian $P$) to be fixed, the first variation of $ E $ with respect to $\cC$ is
$
    \delta  E
    =
    \bra
        P,
        \delta(\cC^{\rT}\cC)
    \ket
    =
    \bra
        P,
        (\delta\cC)^{\rT}\cC+\cC^{\rT}\delta\cC
    \ket
    =
    2
    \bra
        \cC P,
        \delta\cC
    \ket
$,
which implies that
\begin{equation}
\label{dEdC}
    \d_{\cC} E  = 2\cC P.
\end{equation}
Now, substitution of (\ref{dEdA})--(\ref{dEdC}) into (\ref{dEdGamma_blocks}) yields (\ref{dEdGamma}).
\end{proof}
We will now take into account the dependence of the closed-loop system matrices $\cA$, $\cB$, $\cC$ in (\ref{cABC}) on the controller matrices $a$, $b$, $c$, with the latter still considered to be independent variables. In what follows, the Gramians $P$ and $Q$ in (\ref{PQ}) and the Hankelian $H$, defined by
(\ref{H}), inherit the four-block structure of the matrix $\cA$ from
(\ref{cABC}). Their blocks have size $(n\x n)$ and are numbered as
follows:
\begin{equation}
\label{blocks}
    H
    :=
    {\begin{array}{cc}
    {}_{\leftarrow n \rightarrow}  {}_{\leftarrow n\rightarrow} &\\
            {\small\begin{bmatrix}
                H_{11} & H_{12}\\
                H_{21} & H_{22}
            \end{bmatrix}}
    &\!\!\!\!\!
        \begin{matrix}
            \updownarrow\!{}^n\\
            \updownarrow\!{}_n
        \end{matrix}\\
        {}
    \end{array}}
    =
    {\begin{array}{cc}
    {}_{\leftarrow n \rightarrow}  {}_{\leftarrow n\rightarrow} &\\
            {\small\begin{bmatrix}
                H_{\bullet 1} & H_{\bullet 2}
            \end{bmatrix}}
    &\!\!\!\!\!
            \updownarrow^{2n}
        \\ {}
    \end{array}}
    =
    {\begin{array}{cc}
    {}_{\leftarrow 2n \rightarrow}\\
            {\small\begin{bmatrix}
                H_{1\bullet}\\
                H_{2\bullet}
            \end{bmatrix}}
    &\!\!\!\!\!
        \begin{matrix}
            \updownarrow\!{}^n\\
            \updownarrow\!{}_n
        \end{matrix}\\
        {}
    \end{array}}.
\end{equation}
The block $(\cdot)_{11}$ is related to the state variables of the
plant, while $(\cdot)_{22}$ pertains to those of the controller.
The blocks of the matrix $H$ in (\ref{blocks}) are expressed in terms of the block rows of $Q$ and block columns of $P$ as
$    H_{jk}
    =
    Q_{j\bullet}
    P_{\bullet k}
$.

\begin{lemma}
\label{lem:dEdgamma}
Suppose the matrix $\cA$ in (\ref{cABC}) is Hurwitz. Then the Frechet derivative $\d_{\gamma} E
    =
    \begin{bmatrix}
        \d_a E  & \d_b E \\
        \d_c E  & 0
    \end{bmatrix}
$ of $ E $ from (\ref{E})  with respect to the matrix $\gamma$ from (\ref{Gamma}) is
\begin{equation}
\label{dEdgamma}
    \d_{\gamma} E
    =
    2
    \begin{bmatrix}
        H_{22} &     H_{21}\bC^{\rT} + Q _{2\bullet}  \cB\bD^{\rT}\\
            B_2^{\rT}H_{12}    +    D_0^{\rT}  \cC        P _{\bullet 2} & 0
    \end{bmatrix},
\end{equation}
where the matrices $\Gamma_1$, $\Gamma_2$ are defined by (\ref{Gamma0_Gamma12}); $H$,  $P$, $Q$ are given by (\ref{H})--(\ref{PQ}), and the notation (\ref{blocks}) is used.
\end{lemma}

\begin{proof}
Since
$ E $ is a composite function of $a$, $b$, $c$ which enter
(\ref{E}) through the closed-loop system matrices $\cA$,
$\cB$, $\cC$, the chain rule gives
\begin{equation}
\label{dLdgamma}
    \d_{\gamma}  E
    =
    (
        \d_{\gamma} \Gamma
    )^{\dagger}
    (
        \d_{\Gamma}  E
    )
    =
    \bPi
    (
        \Gamma_1^{\rT}
        \d_{\Gamma}  E
        \Gamma_2^{\rT}
    ).
\end{equation}
Here, $(\cdot)^{\dagger}$ is the adjoint in the
sense of the Frobenius inner product of matrices, and
$\bPi$ is the orthogonal projection onto the subspace
$\bGamma$ defined by (\ref{bGamma})--(\ref{bPi}).
Indeed, the first variation of the affine map $\gamma\mapsto \Gamma$, defined by
(\ref{Gamma})--(\ref{Gamma0_Gamma12}), is given by $\delta \Gamma = \Gamma_1(\delta \gamma) \Gamma_2$, which implies that $\d_{\gamma}\Gamma = \[[[\Gamma_1, \Gamma_2\]]]$. Hence,
$
    \delta  E
    =
    \bra
        \d_{\Gamma}  E,
        \delta\Gamma
    \ket
    =
    \bra
        \d_{\Gamma}  E,
        \Gamma_1
        \delta\gamma
        \Gamma_2
    \ket
    =
    \bra
         \Gamma_1^{\rT}
        \d_{\Gamma}  E
        \Gamma_2^{\rT},
        \delta\gamma
    \ket
    =
    \bra
        \bPi
        (
             \Gamma_1^{\rT}
            \d_{\Gamma}  E
            \Gamma_2^{\rT}
        ),
        \delta\gamma
    \ket
$,
which establishes (\ref{dLdgamma}). Substitution of the matrices
$\Gamma_1$ and $\Gamma_2$ from (\ref{Gamma0_Gamma12}) and $\d_{\Gamma} E $ from (\ref{dEdGamma}) into  the right-hand side of (\ref{dLdgamma}) yields
\begin{align*}
    \d_{\gamma} E
    &=
    2
    \bPi
    \left(
        \begin{bmatrix}
            0 & I_n & 0\\
            B_2^{\rT} & 0 & D_0^{\rT}
        \end{bmatrix}
    \begin{bmatrix}
        H & Q\cB\\
        \cC P & 0
    \end{bmatrix}
        \begin{bmatrix}
            0 & \bC^{\rT}\\
            I_n & 0 \\
            0 & \bD^{\rT}
        \end{bmatrix}
        \right)\\
    &=
    2
    \begin{bmatrix}
        H_{22} &     H_{21}\bC^{\rT} + Q _{2\bullet}  \cB\bD^{\rT}\\
            B_2^{\rT}H_{12}    +    D_0^{\rT}  \cC        P _{\bullet 2} & 0
    \end{bmatrix},
\end{align*}
where
Lemma~\ref{lem:dEdGamma} and the notation (\ref{blocks}) are also used,
which proves (\ref{dEdgamma}).
\end{proof}
Finally, we will utilize the Hamiltonian parameterization (\ref{gamma}),  which makes
$ E $ a function of the matrices $R$ and $b$; see Fig.~\ref{fig:PR}.

\begin{theorem}
\label{th:optimal}
A physically realizable stabilizing controller, with Hamiltonian parameterization (\ref{gamma}),  is a critical point of the LQG cost $E$ from (\ref{E}) if and only if there exists a real antisymmetric matrix $\Phi$ of order $n$ such that
\begin{align}
\label{dEdR=0}
    H_{22}
    &=
    -\Phi J_0,\\
\nonumber
 \fM(b)
    +
    H_{21} \bC^{\rT}
    +
    Q_{21} B \bD^{\rT} \ \ \,\quad\qquad&\\
\label{dEdb=0}
    +
    J_0(H_{12}^{\rT} B_2 + P_{21}C_0^{\rT}D_0)J_2\bI^{\rT}
    &=0.
\end{align}
Here,
\begin{equation}
\label{fM}
    \fM
    :=
    \[[[
        \Phi,J
        \mid
        Q_{22},\bD\bD^{\rT}
        \mid
        J_0P_{22}J_0,\bI J_2D_0^{\rT}D_0 J_2 \bI^{\rT}
    \]]]
\end{equation}
is a self-adjoint operator of grade three in the sense of (\ref{cLsum}).
\end{theorem}


\begin{proof}
In view of (\ref{gamma}), the symmetric matrix  $R$ enters the controller only through $a$. Hence,
\begin{equation}
\label{dEdR}
    \d_R E
    =
    (-J_0 \d_a E +(-J_0 \d_a E )^{\rT})/2
    =
    H_{22}^{\rT}J_0
    -
    J_0 H_{22},
\end{equation}
where the relation $\d_a  E  = 2H_{22}$ from Lemma~\ref{lem:dEdgamma} is used.
Unlike $R$, the matrix $b$ both enters $a$ and completely parameterizes $c$, and hence,
\begin{align}
\nonumber
    \rd E /\rd b
    =&
    ((\d_a E ) J_0 +J_0(\d_a E )^{\rT})bJ/2
     + \d_b E\\
\nonumber
     &+ J_0(\d_c E )^{\rT}J_2\bI^{\rT}\\
\nonumber
     =&
    (H_{22} J_0 +J_0H_{22}^{\rT})bJ
    + 2(    H_{21} \bC^{\rT} + Q _{2\bullet}  \cB\bD^{\rT})\\
\label{dEdb}
    &+ 2J_0(    B_2^{\rT}H_{12} + D_0^{\rT}  \cC        P _{\bullet 2})^{\rT}J_2\bI^{\rT},
\end{align}
where (\ref{dEdgamma}) of Lemma~\ref{lem:dEdgamma} is used again.
By introducing a real antisymmetric matrix
\begin{equation}
\label{Phi}
    \Phi
    :=
    (H_{22} J_0 +J_0H_{22}^{\rT})/2,
\end{equation}
and recalling (\ref{cABC}), (\ref{bB_CD}) and (\ref{blocks}), it follows from (\ref{dEdb}) that
\begin{align*}
    (\rd  E /\rd b)/2
     = &
    \Phi bJ
     + H_{21} \bC^{\rT}   +    Q_{21} B \bD^{\rT} + Q_{22} b \bD\bD^{\rT}\\
     &+ J_0(    H_{12}^{\rT}B_2 + P_{21} C_0^{\rT} D_0)J_2\bI^{\rT}\\
     &+ J_0P_{22} J_0 b \bI J_2 D_0^{\rT} D_0J_2\bI^{\rT}\\
    =&
    H_{21} \bC^{\rT}   +    Q_{21} B \bD^{\rT}\\
    &+ J_0(    H_{12}^{\rT}B_2 + P_{21} C_0^{\rT} D_0)J_2\bI^{\rT} + \fM(b),
\end{align*}
where (\ref{c}) and (\ref{fM}) are also used. Therefore, $\rd E /\rd b=0$ is equivalent to (\ref{dEdb=0}). The definition
(\ref{Phi}), which is considered as an equation with respect to $H_{22}$,  determines uniquely   the skew-Hamiltonian part $-\Phi J_0$ of $H_{22}$, so that $H_{22}$ can be represented as
\begin{equation}
\label{H22PhiPsi}
    H_{22} = (\Psi-\Phi)J_0,
\end{equation}
where
\begin{equation}
\label{Psi}
    \Psi
    :=
    (J_0H_{22}^{\rT}-H_{22}J_0)/2
\end{equation}
is a real symmetric matrix of order $n$. Direct comparison of (\ref{Psi}) with (\ref{dEdR}) yields
\begin{equation}
\label{dEdRPsi}
    \d_R  E
    =
    -2J_0\Psi J_0.
\end{equation}
Hence, $\d_R E  = 0$ holds if and only if $\Psi=0$, in which case, (\ref{H22PhiPsi}) takes the form of (\ref{dEdR=0}). Therefore, the property that the controller is a critical point of $ E $ (that is, $\d_R E  = 0$ and $\rd  E /\rd b=0$) is indeed equivalent to the fulfillment of (\ref{dEdR=0}) and (\ref{dEdb=0}) for a real antisymmetric matrix $\Phi$ of order $n$. \end{proof}


For a given matrix $b$ in the Hamiltonian parameterization (\ref{gamma}) of the controller,  (\ref{Psi}) defines a map $\bR(b)\ni R\mapsto \Psi\in \mS_n$ on the set
\begin{equation}
\label{bR}
    \bR(b)
    :=
    \{
        R\in \mS_n:\
        \cA\
        {\rm is\ Hurwitz}
    \}.
\end{equation}
In view of (\ref{dEdRPsi}),  the Frechet derivative of this map with respect to $R$ is expressed in terms of the second order Frechet derivative of the LQG cost of the closed-loop system as
\begin{equation}
\label{dPsidR}
    \d_R\Psi
    =
    -\frac{1}{2}\[[[J_0, J_0\]]]
    \d_R^2 E,
\end{equation}
where we have also used the property that $\[[[J_0, J_0\]]]$ is involutory since
$\[[[J_0, J_0\]]]^2 \!=\! \[[[J_0^2, J_0^2\]]] \!=\! \[[[-I, -I\]]]$ is the identity operator.

\section{A Quasi-separation principle}\label{sec:separation}

The operator $\fM$, which is defined by (\ref{fM}) and acts on the controller gain matrix $b$ from (\ref{bB_CD}), can be partitioned as
\begin{equation}
\label{MMM}
    \fM(b)
    =
    \begin{bmatrix}
        \fM_1(b_1) &
        \fM_2(b_2)
    \end{bmatrix}
\end{equation}
into two operators acting separately on the submatrices $b_1$ and $b_2$. Here,
\begin{align}
\label{fM1}
    \fM_1
    :=&
    \[[[
        \Phi,J_2
        \mid
        Q_{22},I
        \mid
        J_0P_{22}J_0,J_2D_0^{\rT}D_0 J_2
    \]]], \\
\label{fM2}
    \fM_2
    :=&
    \[[[
        \Phi,DJ_1D^{\rT}
        \mid
        Q_{22},DD^{\rT}
    \]]]
\end{align}
are self-adjoint operators of grades three and two. This allows
the equation (\ref{dEdb=0}) for $\rd E /\rd b = 0$ to be split into
\begin{align}
\label{dEdb1=0}
    \fM_1(b_1)
    +
    Q_{21} B_2
    +
    J_0(H_{12}^{\rT} B_2 + P_{21}C_0^{\rT}D_0)J_2
    &= 0,   \\
\label{dEdb2=0}
    \fM_2(b_2)
    +
    H_{21} C^{\rT}
    +
    Q_{21} B_1 D^{\rT} &= 0,
\end{align}
which are equivalent to $\rd E /\rd b_1 = 0$ and $\rd E /\rd b_2 = 0$.
Note that (\ref{dEdb1=0}) corresponds to the equation for the state-feedback matrix
\begin{equation}
\label{chat}
    \wh{c}
    =
    -(D_0^{\rT}D_0)^{-1}
    (B_2^{\rT}\wh{Q}_1 + D_0^{\rT}C_0)
\end{equation}
of the standard LQG controller for the subsidiary classical plant (\ref{class_plant}),
while  (\ref{dEdb2=0}) corresponds to the equation for the Kalman filter observation gain
matrix of the controller
\begin{equation}
\label{b2hat}
    \wh{b}_2
    =
        (\wh{P}_1 C^{\rT}+ B_1D^{\rT})(DD^{\rT})^{-1}.
\end{equation}
Here, it is assumed that the matrix $D_0$ is of full column rank, and $D$ is of full row rank. The matrices $\wh{c}$ and $\wh{b}_2$ from (\ref{chat}) and (\ref{b2hat}) determine the dynamics matrix of the standard LQG controller as $
    \wh{a}
     :=
    A -\wh{b}_2C+B_2\wh{c}$
and
are expressed in terms of the stabilizing
solutions $\wh{Q}_1$, $\wh{P}_1$ of the independent control and filtering algebraic Riccati equations (AREs):
\begin{align*}
    A^{\rT} \wh{Q}_1 & +  \wh{Q}_1 A + C_0^{\rT}C_0 \\
    &=
    (\wh{Q}_1 B_2+ C_0^{\rT}D_0)
    (D_0^{\rT}D_0)^{-1}
    (\wh{Q}_1 B_2+ C_0^{\rT}D_0)^{\rT},\\
    A \wh{P}_1 &+ \wh{P}_1 A^{\rT} + B_1B_1^{\rT} \\
    & =
    (\wh{P}_1 C^{\rT}+ B_1D^{\rT})(DD^{\rT})^{-1}
    (\wh{P}_1 C^{\rT}+ B_1D^{\rT})^{\rT}.
\end{align*}
%
%
The fact, that (\ref{dEdb1=0}) and (\ref{dEdb2=0}) are independent linear equations with respect to $b_1$ and $b_2$, as well as the original partition (\ref{MMM}), can be interpreted as an  analogue of the classical LQG control/filtering separation principle for the CQLQG problem.
In turn, each of the operators $\fM_k$ from (\ref{fM1}) and (\ref{fM2}) can be split into the sum of
self-adjoint  operators $\fM_k^{\diamond}$ and $\fM_k^+$ of grades one and less one:
\begin{align}
\label{fM1_sign}
    \fM_1
    :=&
    \overbrace{
    \[[[
        \Phi,J_2
    \]]]}^{\fM_1^{\diamond}}
    +
    \overbrace{
    \[[[
        Q_{22},I
        \mid
        J_0P_{22}J_0,J_2D_0^{\rT}D_0 J_2
    \]]]}^{\fM_1^+}, \\
\label{fM2_sign}
    \fM_2
    :=&
    \underbrace{
    \[[[
        \Phi,DJ_1D^{\rT}
    \]]]}_{\fM_2^{\diamond}}
    +
    \underbrace{
    \[[[
        Q_{22},DD^{\rT}
    \]]]}_{\fM_2^+}.
\end{align}
By applying Lemma~\ref{lem:spec}, it follows that the spectrum of $\fM_k^{\diamond}$ is symmetric about the origin, while $\fM_k^+\succcurlyeq 0$. Moreover, if $Q_{22} \succ 0$, or $P_{22}\succ 0$ and $D_0$ in (\ref{cZ}) is of full column rank,  then $\fM_1^+\succ 0$. Indeed, the fulfillment of at least one of these conditions implies positive definiteness of at least one of  the positive semi-definite operators on the right-hand side of the representation
\begin{equation}
\label{fM1+}
    \fM_1^+
    =
    \[[[
        Q_{22},I
    \]]]
    +
    \[[[
        J_0P_{22}J_0^{\rT},J_2D_0^{\rT}D_0 J_2^{\rT}
    \]]]
\end{equation}
which follows from $J_0$ and $J_2$ being antisymmetric matrices.
Similarly, the conditions that $Q_{22}\succ 0$ and $D$ is of full row rank ensure that $\fM_2^+ \succ 0$.
In particular, by adapting  \cite[Lemma~5 of Section~VIII]{VP_2010}, it follows that if, in addition to the rank conditions on $D_0$ and $D$,   the controller state-space realization is minimal, then $Q_{22}\succ 0$ and $P_{22}\succ 0$ and hence, $\fM_1^+\succ 0$ and $\fM_2^+\succ 0$.
Therefore, in the cases discussed above, the invertibility of the operators $\fM_1$ and $\fM_2$ in (\ref{fM1_sign})--(\ref{fM2_sign}) can only be destroyed by the presence of the indefinite operators $\fM_1^{\diamond}$ and $\fM_2^{\diamond}$ if the matrix $\Phi$ is large enough  compared to $Q_{22}$. This can be formulated in terms of the matrix
\begin{equation}
\label{dom}
    \Delta
    :=
    Q_{22}^{-1}\Phi
\end{equation}
whose spectrum is pure imaginary and symmetric about zero.

\begin{lemma}
\label{lem:dom}
Suppose the matrix $D$ in (\ref{y}) is of full row rank and $Q_{22}\succ 0$. Also, suppose the spectral radius of the matrix $\Delta$ from (\ref{dom}) satisfies $\br(\Delta)<1$. Then the operators $\fM_1$ and $\fM_2$ in (\ref{fM1}) and (\ref{fM2}) are positive definite.
\end{lemma}
\begin{proof}
Since $\[[[ J_0 P_{22}J_0, J_2D_0^{\rT}D_0 J_2\]]]\succcurlyeq 0$, and $\[[[Q_{22}, I\]]]\succ 0$ (in view of the assumption $Q_{22}\succ 0$), then (\ref{fM1_sign}) and (\ref{fM1+}) imply that
\begin{equation}
\label{lower1}
    \fM_1
    \succcurlyeq
    \fM_1^{\diamond}
    +
    \[[[
        Q_{22}, I
    \]]]
    \succcurlyeq
    (1-\br(\Delta))
    \[[[
        Q_{22}, I
    \]]].
\end{equation}
Here, we use the relation
$
    \br(\[[[Q_{22},I\]]]^{-1} \fM_1^{\diamond})
    =
    \br(\Delta)\br(J_2)= \br(\Delta)
$ which follows from (\ref{rrr}) and the property that the eigenvalues of the canonical antisymmetric matrix $J_2$ are $\pm i$. Therefore, if $\br(\Delta)<1$, then (\ref{lower1}) implies that $\fM_1\succ 0$. By a similar reasoning, under the additional assumption that $D$ is of full row rank (that is, $DD^{\rT}\succ 0$),  it follows from (\ref{fM2_sign}) and (\ref{dom}) that
$
    \fM_2
    \succcurlyeq
    (1-\br(\Delta))
    \fM_2^+\succ 0.
$
Indeed,
$
    \br((\fM_2^+)^{-1}\fM_2^{\diamond})
    =
    \br(\Delta)
    \br(DJ_1D^{\rT}(DD^{\rT})^{-1})
    \< \br(\Delta)
$
since $-I\preccurlyeq iJ_1\preccurlyeq I$ and the Hermitian matrix $(DD^{\rT})^{-1/2}D(iJ_1)D^{\rT}(DD^{\rT})^{-1/2}$ has all its spectrum in $[-1,1]$, so that $\br(DJ_1D^{\rT}(DD^{\rT})^{-1}) \< 1$.
\end{proof}

Assuming invertibility of the operators $\fM_1$ and $\fM_2$ (for example, the fulfillment of conditions of Lemma~\ref{lem:dom} that ensure a stronger property -- positive definiteness  of these  operators), the equations (\ref{dEdb1=0}) and (\ref{dEdb2=0}) can be written more explicitly for $b_1$ and $b_2$:
\begin{align}
\label{b1_sol}
    b_1
    &=
    -
    \fM_1^{-1}
    (Q_{21} B_2
    +
    J_0(H_{12}^{\rT} B_2 + P_{21}C_0^{\rT}D_0)J_2),\\
\label{b2_sol}
    b_2
    &=
    -
    \fM_2^{-1}
    (
    H_{21} C^{\rT}
    +
    Q_{21} B_1 D^{\rT}).
\end{align}
These two  equations are, in principle, amenable to further reduction (to be discussed elsewhere) and will be utilized as assignment operators in the iterative procedure of Section~\ref{sec:newton} for finding the optimal controller.


\section{Second order condition for optimality}\label{sec:d2EdR2}

A second order necessary condition for optimality of the controller with respect to the matrix $R$ of the Hamiltonian parameterization (\ref{gamma}) is the positive semi-definiteness $\d_R^2 E\succcurlyeq 0$ of the appropriate second Frechet derivative of the LQG cost (\ref{E}). Moreover, the positive definiteness $\d_R^2 E\succ 0$ is sufficient for the local strict optimality. To compute the self-adjoint operator $\d_R^2 E$, which acts on the subspace $\mS_n$ of real symmetric matrices of order $n$,  we define a linear operator  $\cJ:\mS_n\to \mR^{2n\x 2n}$ as an appropriate restriction of the grade one linear operator relating $\cA$ with $R$:
\begin{equation}
\label{cJ}
    \cJ
    :=
    \left.
    \[[[
    \begin{bmatrix}
        0_n\\
        J_0
    \end{bmatrix},
    \begin{bmatrix}
        0_n & I_n
    \end{bmatrix}
    \]]]
    \right|_{\mS_n}.
\end{equation}
Its adjoint is
$\cJ^{\dagger}
    =
    -
    \cS
    \[[[
    \begin{bmatrix}
        0_n &
        J_0
    \end{bmatrix},
    \begin{bmatrix}
        0_n \\ I_n
    \end{bmatrix}
    \]]]
$, since $J_0$ is antisymmetric, with $\cS:\mR^{n\x n}\to \mS_n$ the symmetrizer defined by (\ref{cS}).

\begin{lemma}
\label{lem:dEdR2}
Suppose the matrix $\cA$ in (\ref{cABC}) is Hurwitz. Then the second Frechet derivative of $ E $ from (\ref{E})  with respect to the matrix $R$ from (\ref{gamma}) is
\begin{equation}
\label{dEdR2}
    \d_R^2  E
    =
    4\cJ^{\dagger}(\cQ \cL_{\cA} \cS \cP + \cP \cL_{\cA^{\rT}} \cS \cQ)\cJ.
\end{equation}
Here, $\cL_A$ and $\cS$ are the inverse Lyapunov operator and symmetrizer from (\ref{ILO}), (\ref{cS}), and
$\cQ := \[[[Q, I\]]]$ and $\cP:= \[[[I, P\]]]$ are grade one self-adjoint operators (see Section~\ref{sec:class}) of the left and right multiplication by the observability and controllability Gramians $Q$ and $P$ of the closed-loop system from (\ref{PQ}).
\end{lemma}
\begin{proof}
The matrix $R$ only enters the cost $E$ through the matrix $\cA$ of the closed-loop system, and $\cA$ depends affinely on $R$, with $\d_R \cA=\cJ$ the constant operator from (\ref{cJ}). Hence, (\ref{dEdR2}) follows from $\d_R^2 E = \cJ^{\dagger}\d_{\cA}^2 E\cJ$ and  Lemma~\ref{lem:dEdA2} of \ref{sec:d2EdA2}.
\end{proof}

From (\ref{dEdR2}), it follows that  the ``matrix'' representation of the self-adjoint operator $\d_R^2 E$ on the space $\mS_n$ is described by
$$
    \vech(\d_R^2E(M))
    =
    4\Ups^{\rT} (\Omega + \Omega^{\rT})\Ups \vech(M),
$$
where $\vech(M)$ denotes the half-vectorization of a matrix $M \in \mS_n$, that is, the column-wise vectorization of its triangular part below (and including) the main diagonal. Here, the square matrix
$$
    \Omega
    :=
    -(I_{2n}\ox Q)(I_{2n}\ox \cA +\cA\ox I_{2n})^{-1}
    \Sigma (P\ox I_{2n})
$$
of order $4n^2$ represents the operator $\cQ \cL_{\cA} \cS \cP$ on $\mR^{2n\x 2n}$, with $\Sigma$ corresponding to the symmetrizer $\cS: \mR^{2n\x 2n}\to \mS_{2n}$. Also,
$$
    \Ups
    :=
    \left(
        \begin{bmatrix}
            0_n \\ I_n
        \end{bmatrix}
        \ox
        \begin{bmatrix}
            0_n\\
            J_0
        \end{bmatrix}
    \right)
    \Lambda
$$
is a $(4n^2\x n(n+1)/2)$-matrix which represents the operator $\cJ$, defined by (\ref{cJ}), with $\Lambda\in \mR^{n^2\x n(n+1)/2}$ the ``duplication'' matrix  \cite{M_1988,SIG_1998} which expresses the full vectorization of a matrix $M \in \mS_n$ in terms of its half-vectorization by $\vec(M) = \Lambda \vech(M)$.

\section{A Newton-like scheme}\label{sec:newton}

The equations (\ref{b1_sol})--(\ref{b2_sol}) can be combined with iterations for solving the equation $\Psi=0$ for the matrix $\Psi$ from (\ref{Psi}), which is equivalent to the stationarity of the LQG cost $E$ with respect to the matrix $R$ of the Hamiltonian parameterization. The latter part of the scheme, aimed at finding a  root $R\in \bR(b)$ of the equation $\Psi = 0$ from the set (\ref{bR}), can be organized in the form of Newton-Raphson iterations
\begin{equation}
\label{RR}
    R
    \mapsto
    R - (\d_R \Psi)^{-1}(\Psi)
    =
    R - (\d_R^2 E)^{-1}(\d_R E).
\end{equation}
Here, the symmetric matrices $\d_R E$ and $\Psi$ are related by (\ref{dEdRPsi}), and, in view of  (\ref{dPsidR}), the inverse of the operator $\d_R \Psi$ is given by
\begin{equation}
\label{invdPsidR}
    (\d_R\Psi)^{-1}
    =
    -2(\d_R^2 E)^{-1} \[[[J_0, J_0\]]],
\end{equation}
where we have again used the involutional property of the operator $\[[[J_0, J_0\]]]$, and the second order Frechet derivative $\d_R^2 E$ is provided by Lemma~\ref{lem:dEdR2}. If the local strict optimality condition $\d_R^2E\succ 0$ is satisfied, this ensures well-posedness of the inverse in (\ref{invdPsidR}). Thus the equations  (\ref{b1_sol})--(\ref{b2_sol}), considered as assignment operators for $b_1$ and $b_2$, and (\ref{RR}) for $R$,  constitute a Newton-like iterative scheme for numerical computation of the state-space realization matrices of the optimal CQLQG controller.
These three assignment operators are alternated with updating the Gramians of the closed-loop system via the appropriate Lyapunov equations in (\ref{PQ}). The order of this alternation will influence the overall convergence rate of the scheme and is an important computational issue to be explored. Another issue to be taken into account is that the asymptotic stability of the closed-loop system matrix $\cA$ can be violated by the update of the matrices $b_1$, $b_2$, $R$ after which the next iteration becomes impossible. Therefore, being a local optimization algorithm,  the proposed scheme requires a ``stability recovery'' block.
A salient feature of such an algorithm (which is currently under development) is that it involves the inversion of  special self-adjoint operators on matrices which, in general, can only be carried out  via the vectorization of matrices  mentioned in Sections~\ref{sec:class} and \ref{sec:d2EdR2}.

\section{Conclusion}

We have obtained equations for the optimal controller in the Coherent Quantum LQG problem by direct Frechet differentiation of the LQG cost with respect to the pair of matrices which govern the Hamiltonian parameterization of physically realizable quantum controllers.
We have investigated spectral properties of special self-adjoint operators whose inverse plays an important role in the equations and can only be carried out by using  matrix vectorization.
We have established a partial decoupling of these equations with respect to the gain matrices of the optimal controller, which can be interpreted as a quantum analogue of the standard LQG control/filtering  separation principle.
Using this quasi-separation property, we have outlined a Newton-like iterative scheme for numerical computation of the quantum controller. The scheme involves a yet-to-be-explored  freedom of choosing the order in which to perform iterations with respect to the Hamiltonian and gain matrices of the controller to optimize the convergence rate.
The existence and uniqueness of solutions to the equations for the state-space realization matrices of the optimal CQLQG controller also remains an open problem and so does their further reducibility. This circle of  questions is a subject of ongoing research and will be tackled in subsequent publications.

\section*{Acknowledgement}

The work is supported by the Australian Research Council.



%
%
%
%
%
%


\appendix

\section{Invertibility of grade two operators}
\label{sec:inversion}

\begin{lemma}
\label{lem:cLsum2}
Let $r = 2$ in (\ref{cLsum}), and let both matrices $\alpha_1$ and $\beta_1$ be nonsingular. Then the operator $\cM:= \[[[\alpha_1, \beta_1 \mid \alpha_2, \beta_2\]]]$ is invertible if and only if the eigenvalues $\lambda_1, \ldots, \lambda_p$ of $\alpha_1^{-1} \alpha_2$ and the eigenvalues $\mu_1,\ldots, \mu_q$ of $\beta_2\beta_1^{-1}$ satisfy
\begin{equation}
\label{no-1}
    \lambda_j \mu_k \ne -1
    \quad
    {\rm for\ all}\
    j=1,\ldots, p,\
    k = 1,\ldots, q.
\end{equation}
\end{lemma}
\begin{proof}
If $r=2$, the operator (\ref{cLsum}) can be represented as
$
    \cM
    :=
    \[[[
        \alpha_1, \beta_1
        \mid
        \alpha_2, \beta_2
    \]]]
    =
    \cM_1
    \cM_2
$,
where
$
    \cM_1
    :=
    \[[[
        \alpha_1,
        \beta_1
    \]]]
$ and
$
    \cM_2
    :=
    \[[[
        I, I
        \mid
        \alpha_1^{-1}\alpha_2,
        \beta_2\beta_1^{-1}
    \]]]
$.
The operator $\cM_1$ is invertible in view of the nonsingularity of the matrices $\alpha_1$ and $\beta_1$, with $\cM_1^{-1}=\[[[\alpha_1^{-1},\beta_1^{-1}\]]]$. Hence, the invertibility of $\cM$ is equivalent to that of $\cM_2$. In turn, the operator  $\cM_2$ is invertible if and only if its spectrum $\{1+\lambda_j\mu_k:\ 1\< j\< p,\ 1\< k \< q\}$ does not contain $0$, which is equivalent to  (\ref{no-1}).
\end{proof}

By Lemma \ref{lem:cLsum2},  the nonsingularity of the matrix $
    \sum_{k=1}^{2}
    \beta_k^{\rT}\ox \alpha_k
$
of order $pq$ reduces to a joint property of individual spectra of two matrices of orders $p$ and $q$. This reduction does not hold for $r>2$.

\section{Perturbation of inverse Lyapunov operators}
\label{sec:ILO}

We associate an {\it
inverse Lyapunov operator} $\cL_A$ with a Hurwitz matrix $A \in \mR^{n\x n}$,   so that $\cL_A$  maps a
matrix $M \in \mR^{n\x n}$ to the unique solution  $N$ of the algebraic Lyapunov
equation $A N + NA^{\rT} + M = 0$:
\begin{equation}
\label{ILO}
    \cL_A(M)
    :=
    \int_{0}^{+\infty}
    \re^{At} M \re^{A^{\rT}t}
    \rd t.
\end{equation}
 Its adjoint is
$\cL_A^{\dagger} = \cL_{A^{\rT}}$. Since
$\cL_A$ commutes with the transpose, that is, $\cL_A(M^{\rT})
= (\cL_A(M))^{\rT}$, then it also commutes with a {\it symmetrizer}
$\cS$ defined by
\begin{equation}
\label{cS}
    \cS(M)
    :=
    (M+M^{\rT})/2.
\end{equation}
The operator $\cS: \mR^{n\times n} \to \mS_n$ is the orthogonal projection onto the subspace of real symmetric matrices of order $n$.

\begin{lemma}
\label{lem:dPQ}
The Frechet derivatives of the controllability and observability Gramians $P$ and $Q$ of an asymptotically  stable system $(A,B,C)$ with respect to the matrix
$
    \Gamma
    :=
    \begin{bmatrix}
        A & B\\
        C & 0
    \end{bmatrix}
$
are expressed in terms of (\ref{ILO}) and (\ref{cS}) as
\begin{align}
\label{dPdQ1}
    \d_{\Gamma}P
  &=
    2\cL_A\cS
    \[[[
            \begin{bmatrix}
                I & 0
            \end{bmatrix},
            \begin{bmatrix}
                P \\
                B^{\rT}
            \end{bmatrix}
    \]]],\\
\label{dPdQ2}
    \d_{\Gamma} Q
 &=
    2\cL_{A^{\rT}}\cS
    \[[[
            \begin{bmatrix}
                Q & C^{\rT}
            \end{bmatrix},
            \begin{bmatrix}
                I \\
                0
            \end{bmatrix}
    \]]].
\end{align}
\end{lemma}
\begin{proof}
The Frechet differentiability of $P$ and $Q$ is ensured by the assumption that $A$ is Hurwitz. The first variation of the algebraic Lyapunov equation $AP+PA^{\rT} + BB^{\rT} = 0$ yields
\begin{align*}
    0
&=
    (\delta A)P + A\delta P
    +
    (\delta P) A^{\rT}
     + P \delta A^{\rT}
    +
    (\delta B) B^{\rT} + B\delta B^{\rT}\\
&=
    A\delta P
    +
    (\delta P) A^{\rT}
    +
    2\cS
    \left(
            \begin{bmatrix}
                \delta A & \delta B
            \end{bmatrix}
            \begin{bmatrix}
                P \\
                B^{\rT}
            \end{bmatrix}
    \right).
\end{align*}
This is an algebraic Lyapunov equation with respect to  $\delta P$ with the same matrix $A$,
which proves (\ref{dPdQ1}) in view of the identity $\begin{bmatrix}A & B\end{bmatrix}=  \begin{bmatrix}I & 0\end{bmatrix} \Gamma$. The relation (\ref{dPdQ2}) is  obtained by a similar reasoning from the first variation of the Lyapunov equation for the observability Gramian $Q$, or  by using the duality between $P$ and $Q$.
\end{proof}

\section{Second Frechet derivative of the LQG cost}
\label{sec:d2EdA2}

\begin{lemma}
\label{lem:dEdGamma2}
The second Frechet derivative of the squared $\cH_2$-norm $E :=\|(A,B,C)\|_2^2$ of an asymptotically stable system with respect to the matrix $
    \Gamma
    :=
    \begin{bmatrix}
        A & B\\
        C & 0
    \end{bmatrix}
$  is computed as
\begin{align}
\nonumber
    \d_{\Gamma}^2  E
    =&
    4
    \[[[
        \begin{bmatrix}
            I\\
            0
        \end{bmatrix},
        \begin{bmatrix}
            P & B
        \end{bmatrix}
    \]]]
    \cL_{A^{\rT}}
    \cS
    \[[[
        \begin{bmatrix}
            Q & C^{\rT}
        \end{bmatrix},
        \begin{bmatrix}
            I\\
            0
        \end{bmatrix}
    \]]]\\
\nonumber
    &+4
    \[[[
        \begin{bmatrix}
            Q\\
            C
        \end{bmatrix},
        \begin{bmatrix}
            I & 0
        \end{bmatrix}
    \]]]
    \cL_A
    \cS
    \[[[
        \begin{bmatrix}
            I & 0
        \end{bmatrix},
        \begin{bmatrix}
            P\\
            B^{\rT}
        \end{bmatrix}
    \]]]\\
\label{dEdGamma2}
    &+
    2
    \[[[
        \begin{bmatrix}
          Q & 0 \\
          0 & I \\
        \end{bmatrix},
        \begin{bmatrix}
          0 & 0 \\
          0 & I \\
        \end{bmatrix}
        \mid
        \begin{bmatrix}
          0 & 0 \\
          0 & I \\
        \end{bmatrix},
        \begin{bmatrix}
          P & 0 \\
          0 & I \\
        \end{bmatrix}
    \]]].
\end{align}
Here, $\cL_A$ and $\cS$ are the inverse Lyapunov operator  and symmetrizer from (\ref{ILO}), (\ref{cS}), and $P$, $Q$ are  the controllability and observability Gramians of the system.
\end{lemma}
\begin{proof}
Lemma~\ref{lem:dEdGamma} implies that the first variation of the Frechet derivative $\d_{\Gamma}E$ is computed as
\begin{align*}
    \delta \d_{\Gamma}E/2
    =&
    \delta
    \begin{bmatrix}
        QP & QB\\
        CP & 0
    \end{bmatrix}\\
    =&
    \begin{bmatrix}
        I\\
        0
    \end{bmatrix}
    \delta Q
    \begin{bmatrix}
        P & B
    \end{bmatrix}
    +
    \begin{bmatrix}
        Q\\
        C
    \end{bmatrix}
    \delta P
    \begin{bmatrix}
        I & 0
    \end{bmatrix}
    +
    \begin{bmatrix}
    0 & Q \delta B\\
    (\delta C) P & 0
    \end{bmatrix}.
\end{align*}
Hence, (\ref{dEdGamma2}) is obtained
by using the Frechet derivatives of the Gramians from Lemma~\ref{lem:dPQ} of \ref{sec:ILO} and the identity
$$
    \begin{bmatrix}
    0 & Q \delta B\\
    (\delta C) P & 0
    \end{bmatrix}
    =
    \begin{bmatrix}
      Q & 0 \\
      0 & I
    \end{bmatrix}
    \delta\Gamma
    \begin{bmatrix}
      0 & 0 \\
      0 & I
    \end{bmatrix}
    +
    \begin{bmatrix}
      0 & 0 \\
      0 & I
    \end{bmatrix}
    \delta\Gamma
    \begin{bmatrix}
      P & 0 \\
      0 & I
    \end{bmatrix}.
$$
\end{proof}

\begin{lemma}
\label{lem:dEdA2}
The second Frechet derivative of the squared $\cH_2$-norm $ E :=\|(A,B,C)\|_2^2$ of an asymptotically stable system with respect to $A$ is
\begin{equation}
\label{dEdA2}
    \d_A^2  E
    =
    4\cR,
    \qquad
    \cR
    :=
    \cQ \cL_A \cS \cP + \cP \cL_{A^{\rT}} \cS \cQ.
\end{equation}
Here,
$    \cQ:= \[[[Q,I\]]]$ and
$
    \cP:= \[[[I,P\]]]
$
are grade one self-adjoint operators (see Section~\ref{sec:class}) of the left and right multiplication by the observability and controllability Gramians  of the system.
\end{lemma}
\begin{proof}
In view of Lemma~\ref{lem:dPQ}, the first variation of $\d_A E  = 2QP$ with respect to $A$ is
\begin{align*}
    \delta \d_A  E
    &=
    2(Q\delta P + (\delta Q) P)\\
    &=
    4(
        Q\cL_A \cS((\delta A) P)
        +
        \cL_{A^{\rT}} \cS(Q(\delta A))P)
\end{align*}
which establishes (\ref{dEdA2}). Alternatively, (\ref{dEdA2}) can be obtained from (\ref{dEdGamma2}) of Lemma~\ref{lem:dEdGamma2}.
\end{proof}
Note that at least some eigenvalues of the self-adjoint operator $\cR$ in (\ref{dEdA2}) are positive, since
$
    \cR(A) = -QP
$ is the negative of the Hankelian,
and
$
    \bra
        A,
        \cR(A)
    \ket
    =
    -
    \bra
        A,
        QP
    \ket = \|(A,B,C)\|_2^2/2> 0
$.

\end{document}